\documentclass[copyright,creativecommons,sharealike]{eptcs}
\providecommand{\event}{FIT 2012} % Name of the event you are submitting to
\usepackage{breakurl}             % Not needed if you use pdflatex only.

\usepackage{wmsty}
\usepackage[sort]{cite}

\title{A Robust Specification Theory \\ for Modal Event-Clock Automata}

\author{Uli Fahrenberg \qquad Axel Legay \quad\mbox{}
  \institute{Irisa / INRIA \\ Rennes, France}}

\newcommand*\authorrunning{U.~Fahrenberg and A.~Legay}
\newcommand*\titlerunning{A Robust Specification Theory for Modal
  Event-Clock Automata}

\begin{document}

\maketitle

\begin{abstract}
  In a series of recent work, we have introduced a general framework for
  quantitative reasoning in specification theories.
  % introduced quantitative modal automata, a quantitative extension of
  % the classical theory of modal specification systems defined proposed
  % by Larsen.
  The contribution of
  this paper is to show how this framework can be applied to yield a
  robust specification theory 
  % quantitative modal automata can be used as a general framework to
  % define a notion of robustness
  for timed specifications.
\end{abstract}

\section{Introduction}

Specification theories allow to reason about behaviors of systems at
the abstract level, which is needed in various application such as
abstraction-based model checking for programming languages, or
compositional reasoning. Depending on the application for which they
are used, such specification theories may come together with (1) a
satisfaction relation that allows to decide whether an implementation
is a model of the specification, (2) a notion of refinement for
determining the relationship between specifications and their set of
implementations, (3) a structural composition which at the abstract
level mimics the behavioral composition of systems, (4) a quotient
that allows to synthesize specifications from refinements, and (5) a
logical composition that allows to compute intersections of sets of
implementations, \cf~\cite{DBLP:conf/fase/BauerDHLLNW12}.

Prominent among existing specification theories, outside logics, is the
one of \emph{modal transition
  systems}~\cite{DBLP:conf/avmfss/Larsen89,DBLP:conf/ictac/BenesKLS09,SCU11,Nyman08Thesis,DBLP:conf/fmoods/GrulerLS08,DBLP:conf/concur/GodefroidHJ01,DBLP:conf/vmcai/GrumbergLLS05}
which are labeled transition systems equipped with two types of
transitions: \must~transitions that are mandatory for any
implementation, and \may~transitions which are optional for an
implementation. So far, existing modal specification theories have
relied on Boolean versions of both the refinement and the satisfaction
relation. They are hence \emph{fragile} in the sense that they are
unable to quantify the impact of small variations of the behavior of the
environment in which a component is working. In a series of recent
work~\cite{journals/mscs/BauerJLLS11,DBLP:conf/mfcs/BauerFJLLT11,conf/csr/BauerFLT12},
and building on a general theory of quantitative analysis of
systems~\cite{journals/jlap/ThraneFL10,journals/cai/FahrenbergLT10,DBLP:journals/tcs/LarsenFT11,DBLP:journals/corr/abs-1107-1205,conf/fsttcs/FahrenbergLT11},
we have leveraged this problem by extending modal specifications from
the Boolean to the quantitative world and introducing truly quantitative
versions of the operators mentioned above.

The contribution of this paper is to show how our general quantitative
framework from~\cite{conf/csr/BauerFLT12} can be used to define a notion
of robustness for timed modal specifications, or model event-clock
specifications (MECS)~\cite{DBLP:conf/icfem/BertrandLPR09}.  We first
observe that the notion of refinement proposed
in~\cite{DBLP:conf/icfem/BertrandLPR09} is not adequate to reason on
MECS in a robust manner. We then propose a new version of refinement
that can capture quantitative phenomena in a realistic manner, and
proceed to exhibit the properties of the above specification-theory
operators with respect to this quantitative refinement.  We show that
structural composition and quotient have properties which are useful
generalizations of their standard Boolean properties, hence they can be
employed for robust reasoning on MECS without problem.  Conjunction, on
the other hand, is generally not robust (similarly to the problems
exposed in~\cite{DBLP:conf/mfcs/BauerFJLLT11}), but together with the
new operator of quantitative widening can be used in a robust manner.

% . Our refinement is based on a max-lead distance and cannot thus be
% handled by the framework in~\cite{DBLP:conf/mfcs/BauerFJLLT11}.  We
% instantiate our framework to get a complete robust specification
% theory for these models, using a system distance which is appropriate
% for this context and cannot be handled by the framework
% in~\cite{DBLP:conf/mfcs/BauerFJLLT11}.

\section{Quantitative Specification Theories}

General quantitative specification theories have been introduced
in~\cite{conf/csr/BauerFLT12}.  These consist of
\begin{itemize}
\item a specification formalism: modal transition systems with labels
  drawn from a set $\Spec$,
\item a distance on traces of labels: $d_T: \Spec\times \Spec\to
  \Realnn$, and
\item operations on specifications which allow high-level reasoning and
  which generally are continuous with respect to the natural distance on
  specifications induced by the trace distance.
\end{itemize}
Below we give a more detailed account of these things, in order to be
able to apply them to modal event-clock specifications later.

\subsection{Structured Modal Transition Systems}

We assume that the set $\Spec$ of labels comes with a partial order
$\sqsubseteq_\Spec$ modeling \emph{refinement} of data: if
$k\sqsubseteq_\Spec \ell$, then $k$ is more refined (leaves fewer
choices) than $\ell$.  The set $\Imp= \{ k\in \Spec\mid
k'\sqsubseteq_\Spec k\Longrightarrow k'= k\}$ is called the set of
\emph{implementation labels}; these are the data which cannot be refined
further.

We let $\llbracket k\rrbracket=\{ k'\in \Imp\mid k'\sqsubseteq k\}$
denote the set of implementation refinements of a label $k$, and we
assume that $\Spec$ is well-formed in the sense that $\llbracket
k\rrbracket\ne \emptyset$ for all $k\in \Spec$: any specification label
can be implemented.

A \emph{structured modal transition system} (SMTS) is a tuple $( S, s_0,
\mmayto_S, \mmustto_S)$ consisting of a set $S$ of states, an initial
state $s_0\in S$, and \must and \may transitions $\mmustto_S,
\mmayto_S\subseteq S\times \Spec\times S$ for which it holds that for
all $s\mustto[ k]_S s'$ there is $s\mayto[ \ell]_S s'$ with
$k\sqsubseteq_\Spec \ell$.  This last condition is one of
\emph{consistency}: everything which is required, is also allowed.

An SMTS $( S, s_0, \mmayto_S, \mmustto_S)$ is an \emph{implementation}
if $\mmustto_S= \mmayto_S\subseteq S\times \Imp\times S$, \ie~an
ordinary labeled transition system with labels in $\Imp$.  Hence in an
implementation, all optional behavior has been resolved, and all data
has been refined to implementation labels.

A \emph{modal refinement} of SMTS $S$, $T$ is a relation $R\subseteq
S\times T$ such that for any $( s, t)\in R$,
\begin{itemize}
\item whenever $s\mayto[ k]_S s'$, then also $t\mayto[ \ell]_T t'$ for
  some $k\sqsubseteq_\Spec \ell$ and $( s', t')\in R$,
\item whenever $t\mustto[ \ell]_T t'$, then also $s\mustto[ k]_S s'$ for
  some $k\sqsubseteq_\Spec \ell$ and $( s', t')\in R$.
\end{itemize}
Thus any behavior which is permitted in $S$ is also permitted in $T$,
and any behavior required in $T$ is also required in $S$.  We write
$S\le_m T$ if there is a modal refinement $R\subseteq S\times T$ with $(
s_0, t_0)\in R$, and $S\equiv_m T$ if there is a two-sided refinement
$S\le_m T$ and $T\le_m S$.

The \emph{implementation semantics} of a SMTS $S$ is the set $\llbracket
S\rrbracket=\{ I\le_m S\mid I\text{ is an implementation}\}$, and we
write $S\le_t T$ if $\llbracket S\rrbracket\subseteq \llbracket
T\rrbracket$, saying that $S$ thoroughly refines $T$.

\subsection{Distances}

The above setting is purely \emph{qualitative}, \ie~Boolean: a
refinement $S\le_m T$ either holds, or it does not; a transition system
$I$ either is an implementation of a specification $S$, or it is not.
In order to turn this setting into a \emph{quantitative} one, where we
can reason about \emph{robustness} of refinements and implementations,
we need to introduce \emph{distances}.

We have in~\cite{conf/fsttcs/FahrenbergLT11} developed a general
framework which allows to reason about a variety of such system
distances in a uniform way.  To apply this to specifications, let
$\Spec^\infty= \Spec^*\cup \Spec^\omega$ denote the set of finite and
infinite traces over $\Spec$, and let $d_T: \Spec^\infty\times
\Spec^\infty\to \Realnn\cup\{ \infty\}$ be an extended hemimetric.
Recall that this means that $d_T( \sigma, \sigma)= 0$ for all $\sigma\in
\Spec^\infty$, and that $d_T( \sigma_1, \sigma_2)+ d_T( \sigma_2,
\sigma_3)\ge d_T( \sigma_1, \sigma_3)$ for all $\sigma_1, \sigma_2,
\sigma_3\in Spec^\infty$.  Note that as $\Spec\subseteq \Spec^\infty$,
$d_T$ induces a hemimetric on $\Spec$.

Let $M$ be an arbitrary set and $\LL=( \Realnn\cup\{ \infty\})^M$ the
set of functions from $M$ to the extended non-negative real line.  Then
$\LL$ is a complete lattice with partial order
$\mathord{\sqsubseteq_\LL}$ given by $\alpha\sqsubseteq_\LL \beta$ if
and only if $\alpha( x)\le \beta( x)$ for all $x\in M$, and with an
addition $\oplus_\LL$ given by $( \alpha\oplus_\LL \beta)( x)= \alpha(
x)+ \beta( x)$.  The bottom element of $\LL$ is also the zero of
$\oplus_\LL$ and given by $\bot_\LL( x)= 0$, and the top element is
$\top_\LL( x)= \infty$.  We also define a metric on $\LL$ by $d_\LL(
\alpha, \beta)= \sup_{ x\in M}| \alpha( x)- \beta( x)|$.

Let $F: \Spec\times \Spec\times \LL\to \LL$ be a function with the
following properties:
\begin{itemize}
\item $F$ is continuous in the first two coordinates: $F( \cdot, k,
  \alpha)$ and $F( k, \cdot, \alpha)$ are continuous functions $\Imp\to
  \LL$ for all $k\in \Spec$, $\alpha\in \LL$.
\item $F$ is monotone in the third coordinate: $F( k, \ell, \cdot)$ is a
  monotone function $\LL\to \LL$ for all $k, \ell\in \Spec$.
\item $F( \cdot, \cdot, \bot_\LL)$ extends $d_T$: for all $k, \ell\in
  Spec$, $F( k, \ell, \bot_\LL)= d_T( k, \ell)$.
\item $F$ acts as a Hausdorff metric~\cite{munkres2000topology} when
  specification labels are viewed as sets of implementation labels: for
  all $k, \ell\in \Spec$ and $\alpha\in \LL$, $F( k, \ell, \alpha)=
  \sup_{ m\in \llbracket k\rrbracket} \inf_{ n\in \llbracket
    \ell\rrbracket} F( m, n, \alpha)$.
\item Sets of implementation labels are closed with respect to $F$: for
  all $k, \ell\in \Spec$ and $\alpha\in \LL$ with $F( k, \ell,
  \alpha)\ne \top_\LL$, there are $m\in \llbracket k\rrbracket$, $n\in
  \llbracket \ell\rrbracket$ with $F( m, \ell, \alpha)= F( k, n,
  \alpha)= F( k, \ell, \alpha)$.
  % \item $F$ satisfies an extended triangle inequality:\uli{More
  %     general version.} for all $k, \ell, m\in \Spec$ and $\alpha,
  %   \beta, \gamma\in \LL$ with $\alpha\oplus_\LL \beta\sqsupseteq_\LL
  %   \gamma$, $F( k, \ell, \alpha)\oplus_\LL F( \ell, m,
  %   \beta)\sqsupseteq_\LL F( k, m, \gamma)$.
\item $F$ satisfies an extended triangle inequality: for all $k, \ell,
  m\in \Spec$ and $\alpha, \beta\in \LL$, $F( k, \ell, \alpha)\oplus_\LL
  F( \ell, m, \beta)\sqsupseteq_\LL F( k, m, \alpha\oplus_\LL \beta)$.
\end{itemize}

As the last ingredients, let $h_T: \Spec^\infty\times \Spec^\infty\to
\LL$ and $g: \LL\to \Realnn\cup\{ \infty\}$ be functions such that $g$
is monotone with $g( \bot_\LL)= 0$, $g( \alpha)\ne \infty$ for
$\alpha\ne\top_\LL$, and $g\circ h_T= d_T$, and such that $h_T$ has a
recursive characterization, using $F$, as follows:
\begin{equation}
  \label{eq:trdist}
  h_T( \sigma, \tau)=
  \begin{cases}
    F( \sigma_0, \tau_0, h_T( \sigma^1, \tau^1)) &\text{if }
    \sigma, \tau\ne \emptyseq, \\
    \top_\LL &\text{if } \sigma= \emptyseq, \tau\ne \emptyseq \text{ or }
    \sigma\ne \emptyseq, \tau= \emptyseq, \\
    \bot_\LL &\text{if } \sigma= \tau= \emptyseq.
  \end{cases}
\end{equation}
Here $\emptyseq\in \Specseq$ denotes the empty sequence, and for any
$\sigma\in \Specseq$, $\sigma_0$ denotes its first element and
$\sigma^1$ the tail of $\sigma$ with the first element removed.

For technical reasons, we will work mostly with the auxiliary function
$h_T: \Spec^\infty\times \Spec^\infty\to \LL$ below instead of the
distance $d_T$; indeed, the framework in~\cite{conf/csr/BauerFLT12} has
been developed completely without reference to the distance $d_T$ which,
from a point of view of applications, should be the actual function of
interest.  This is due to the fact that the recursive characterization
in~\eqref{eq:trdist} needs to ``live'' in $\LL$ to be applicable to
non-trivial distances, \cf~\cite{conf/fsttcs/FahrenbergLT11}.

We assume all SMTS to be \emph{compactly
  branching}~\cite{Breugel95-Metric}, that is, for any SMTS $S$ and any
$s\in S$, the sets $\{ k\in \Spec\mid s\mayto[ k] s'\}$ and $\{ k\in
\Spec\mid s\mustto[ k] s'\}$ are to be compact under the hemimetric
$d_T$.  A SMTS $S$ is said to be \emph{deterministic} if it holds for
all $s\in S$, $s\mayto[ k_1]_S s_1$, $s\mayto[ k_2]_S s_2$ for which
there is $k\in \Spec$ with $h_T( k, k_1)\ne \top_\LL$ and $h_T( k,
k_2)\ne \top_\LL$ that $k_1= k_2$ and $s_1= s_2$.

\subsection{Operations}

Any specification theory comes equipped with certain operations which
allow high-level reasoning~\cite{DBLP:conf/fase/BauerDHLLNW12}:
refinement, structural composition and quotient, and conjunction.  For
our quantitative framework, we add an operation of \emph{widening} which
allows to systematically relax specifications.

The \emph{modal refinement distance} $d_m: S\times T\to \Realnn\cup\{
\infty\}$ between the states of SMTS $S$, $T$ is defined using an
auxiliary function $h_m: S\times T\to \LL$, which in turn is defined to
be the least fixed point to the equations
\begin{equation*}
  h_m( s, t)= \max\left\{
    \begin{aligned}
      & \smashsup_{ s\,\mayto[ k]_S\, s'\,} \inf_{ \, t\,\mayto[
        \ell]_T\, t'} F( k, \ell, h_m( s', t')), \\
      & \smashsup_{ t\,\mustto[ \ell]_T\, t'\,} \inf_{ \, s\,\mustto[
        k]_S\, s'} F( k, \ell, h_m( s', t')).
    \end{aligned}
  \right.
\end{equation*}
We let $d_m= g\circ h_m$, using the function $g: \LL\to \Realnn\cup\{
\infty\}$ from above.  Also, $d_m(S, T) = d_m(s_0, t_0)$, and we write
$S\le_m^\alpha T$ if $d_m( S, T)\sqsubseteq_\LL \alpha$.  This
definition is an extension of the one of \emph{simulation distance}
in~\cite{DBLP:journals/corr/abs-1107-1205}, and the proof of existence of the
least fixed point is similar to the one
in~\cite{DBLP:journals/tcs/LarsenFT11}.  Note also that $d_m$ extends the
refinement relation $\le_m$ in the sense that $s\le_m t$ implies $d_m(
s, t)= 0$.

The \emph{thorough refinement distance} from an SMTS $S$ to an SMTS $T$
is
\begin{equation*}
  d_t( S, T)= \adjustlimits \sup_{ I\in \llbracket S\rrbracket} \inf_{
    J\in \llbracket T\rrbracket} d_m( I, J),
\end{equation*}
and we write $S\le_t^\alpha T$ if $d_t( S, T)\sqsubseteq_\LL \alpha$.
Again, $S\le_t T$ implies $d_t( S, T)= 0$.  It can be
shown~\cite{conf/csr/BauerFLT12} that both $d_m$ and $d_t$ obey triangle
inequalities in the sense that $d_m( S, T)+ d_m( T, U)\ge d_m( S, U)$
and $d_t( S, T)+ d_t( T, U)\ge d_t( S, U)$ for all SMTS $S$, $T$, $U$.
Also, $d_t( S, T)\le d_m( S, T)$ for all SMTS $S$, $T$, and $d_t( S, T)=
d_m( S, T)$ if $T$ is deterministic~\cite{conf/csr/BauerFLT12}.

To introduce \emph{structural composition} and \emph{quotient} of SMTS,
one needs corresponding operators on labels.  Let thus $\obar:
\Spec\times \Spec\hookrightarrow \Spec$ and $\obslash: \Spec\times
\Spec\to \Spec$ be partial label operators which satisfy the following
conditions:
\begin{itemize}
\item For all $k, \ell, k', \ell'\in \Spec$, if $h_T( k, \ell)\ne
  \top_\LL$ and $h_T( k', \ell')\ne \top_\LL$, then $k\obar k'$ is
  defined if and only if $\ell\obar \ell'$ is defined;
\item for all $k, \ell, m\in \Spec$, $\ell\obslash k$ is defined and
  $m\sqsubseteq_\Spec \ell\obslash k$ if and only if $k\obar m$ is
  defined and $k\obar m\sqsubseteq_\Spec \ell$;
\item for all $\ell, \ell'\in \Spec$, the following conditions are
  equivalent:
  \begin{itemize}
  \item there exists $k\in \Spec$ for which both $h_T( k, \ell)\ne
    \top_\LL$ and $d_T( k, \ell')\ne \top_\LL$;
  \item there exists $m\in \Spec$ for which both $\ell\obar m$ and
    $\ell'\obar m$ are defined;
  \item there exists $m\in \Spec$ for which both $m\obslash \ell$ and
    $m\obslash \ell'$ are defined.
  \end{itemize}
\end{itemize}

The \emph{structural composition} of SMTS $S$, $T$ is then the SMTS
$S\| T=( S\times T,( s_0, t_0), \mmayto_{ S\| T}, \mmustto_{ S\| T})$
with transitions defined as follows:
\begin{equation*}
  \dfrac{ s\mayto[ k]_S s' \qquad t\mayto[ \ell]_T t' \qquad k\obar \ell
    \text{ defined}} {( s, t)\mayto[ k\obar \ell]_{ S\| T}( s',
    t')}\qquad \dfrac{ s\mustto[ k]_S s' \qquad t\mustto[ \ell]_T t'
    \qquad k\obar \ell \text{ defined}} {( s, t)\mustto[ k\obar \ell]_{
      S\| T}( s', t')}
\end{equation*}

It can be shown~\cite{journals/mscs/BauerJLLS11} that for all SMTS $S$,
$S'$, $T$, $T'$, $S\le_m T$ and $S'\le_m T'$ imply $S\| S'\le_m T\| T'$.
For a quantitative generalization of this, we need a function $P:
\LL\times \LL\to \LL$ which permits to infer bounds on distances on
synchronized labels.  We assume that $P$ is monotone in both
coordinates, has $P( \bot_\LL, \bot_\LL)= \bot_\LL$, $P( \alpha,
\top_\LL)= P( \top_\LL, \alpha)= \top_\LL$ for all $\alpha\in \LL$, and
that
\begin{equation*}
  F( k\obar k', \ell\obar \ell', P( \alpha, \alpha'))\sqsubseteq_\LL P(
  F( k, \ell, \alpha), F( k', \ell', \alpha'))
\end{equation*}
for all $k, \ell, k', \ell'\in \Spec$ and $\alpha, \alpha'\in \LL$ for
which $k\obar k'$ and $\ell\obar \ell'$ are defined.  Then $P$ can be
used to bound distances between structural compositions: for SMTS $S$,
$T$, $S'$, $T'$, we have $h_m( S\| S', T\| T')\sqsubseteq_\LL P( h_m( S,
T), h_m( S', T'))$~\cite[Thm.~2]{conf/csr/BauerFLT12}.

For the definition of quotient, we first need to introduce
\emph{pruning}.  For a SMTS $S$ and a subset $B\subseteq S$ of states,
the pruning $\rho_B( S)$ is given as follows: Define a \must-predecessor
operator $\pre: 2^S\to 2^S$ by $\pre( S')=\{ s\in S\mid \exists k\in
\Spec, s'\in S': s\mustto[ k] s'\}$ and let $\pre^*$ be the reflexive,
transitive closure of $\pre$.  Then $\rho_B( S)$ exists if $s_0\notin
\pre^*( B)$, and in that case, $\rho_B( S)=( S_\rho, s_0, \mmayto_\rho,
\mmustto_\rho)$ with $S_\rho= S\setminus \pre^*( B)$, $\mmayto_\rho=
\mmayto\cap( S_\rho\times \Spec \times S_\rho)$, and $\mmustto_\rho=
\mmustto\cap( S_\rho\times \Spec \times S_\rho)$.

The \emph{quotient} of an SMTS $T$ by an SMTS $S$ is the SMTS $T\bbslash
S= \rho_B( T\times S\cup\{ u\},( t_0, s_0), \mmayto_{ T\bbslash S},
\mmustto_{ T\bbslash S})$ given as follows (if it exists):
\begin{gather*}
  \dfrac{%
    t\mayto[ \ell]_T t' \qquad s\mayto[ k]_S s' \qquad \ell\obslash k
    \text{ defined}}{%
    ( t, s)\mayto[ \ell\obslash k]_{ T\bbslash S}( t', s')} \qquad
  \dfrac{%
    t\mustto[ \ell]_T t' \qquad s\mustto[ k]_S s' \qquad \ell\obslash k
    \text{ defined}}{%
    ( t, s)\mustto[ \ell\obslash k]_{ T\bbslash S}( t', s')} \\
  \dfrac{%
    t\mustto[ \ell]_T t' \qquad \forall s\mustto[ k]_S s': \ell\obslash k
    \text{ undefined}}{%
    ( t, s)\in B} \\
  \dfrac{%
    m\in \Spec \qquad \forall s\mayto[ k]_S s': k\obar m \text{
      undefined}}{%
    ( t, s)\mayto[ m]_{ T\bbslash S} u} \qquad \dfrac{%
    m\in \Spec}{%
    u\mayto[ m]_{ T\bbslash S} u}
\end{gather*}
Note the extra universal state $u$ which is introduced here.  The
standard property of quotient is as
follows~\cite{journals/mscs/BauerJLLS11}: For SMTS $S$, $T$, $X$, for
which $S$ is deterministic and $T\bbslash S$ exists, $X\le_m T\bbslash
S$ if and only if $S\| X\le_m T$.  Note that this property implies
\emph{uniqueness} (up to $\equiv_m$) of
quotient~\cite{conf/models/FahrenbergLW11}; hence if quotient exists, it
must be defined as above.

For quantitative properties of quotient, we must again look to
properties of the label operator $\obslash$ which can ensure them.  We
say that $\obslash$ is \emph{quantitatively well-behaved} if it holds
for all $k, \ell, m\in \Spec$ that $\ell\obslash k$ is defined and $h_T(
m, \ell\obslash k)\ne \top_\LL$ if and only if $k\obar m$ is defined and
$d_T( k\obar m, \ell)\ne \top_\LL$, and in that case, $F( m,
\ell\obslash k, \alpha)\sqsupseteq_\LL F( k\obar m, \ell, \alpha)$ for
all $\alpha\in \LL$.  For such a quantitatively well-behaved $\obslash$
it can be shown~\cite[Thm.~3]{conf/csr/BauerFLT12} that for all SMTS
$S$, $T$, $X$ such that $S$ is deterministic and $T\bbslash S$ exists,
$h_m( X, T\bbslash S)\sqsupseteq_\LL h_m( S\| X, T)$.

For \emph{conjunction} of SMTS, we need a partial label operator
$\owedge: \Spec\times \Spec\to \Spec$ for which it holds that
\begin{itemize}
\item for all $k, \ell\in \Spec$, if $k\owedge \ell$ is defined, then
  $k\owedge \ell\sqsubseteq_\Spec k$ and $k\owedge \ell\sqsubseteq_\Spec
  \ell$
\item for all $k, \ell, m\in \Spec$ for which $m\sqsubseteq_\Spec k$ and
  $m\sqsubseteq_\Spec \ell$, $k\owedge \ell$ is defined and
  $m\sqsubseteq_\Spec k\owedge \ell$, and
\item for all $\ell, \ell'\in \Spec$, there exists $k\in \Spec$ for
  which $h_T( k, \ell)\ne \top_\LL$ and $h_T( k, \ell')\ne \top_\LL$ if
  and only if there exists $m\in \Spec$ for which $\ell\owedge m$ and
  $\ell'\owedge m$ are defined.
\end{itemize}

The conjunction of two SMTS $S$, $T$ is the SMTS $S\wedge T=\rho_B(
S\times T,( s_0, t_0), \mmayto_{ S\wedge T}, \mmustto_{ S\wedge T})$
given as follows:
\begin{gather*}
  \dfrac{%
    s\mustto[ k]_S s'\qquad t\mayto[ \ell]_T t'\qquad k\owedge \ell\text{
      defined}%
  }{%
    ( s, t)\mustto[ k\owedge \ell]_{ S\wedge T}( s', t')%
  } \qquad%
  \dfrac{%
    s\mayto[ k]_S s'\qquad t\mustto[ \ell]_T t'\qquad k\owedge \ell\text{
      defined}%
  }{%
    ( s, t)\mustto[ k\owedge \ell]_{ S\wedge T}( s', t')%
  } \\%
  \dfrac{%
    s\mayto[ k]_S s'\qquad t\mayto[ \ell]_T t'\qquad k\owedge \ell\text{
      defined}%
  }{%
    ( s, t)\mayto[ k\owedge \ell]_{ S\wedge T}( s', t')%
  } \\%
  \dfrac{%
    s\mustto[ k]_S s'\qquad \forall t\mayto[ \ell]_T t': k\owedge
    \ell\text{ undefined}%
  }{%
    ( s, t)\in B} \qquad%
  \dfrac{%
    t\mustto[ \ell]_T t'\qquad \forall s\mayto[ k]_S s': k\owedge
    \ell\text{ undefined}%
  }{%
    ( s, t)\in B}
\end{gather*}

With this definition, it can be shown~\cite{journals/mscs/BauerJLLS11}
that conjunction acts as \emph{greatest lower bound}: Given SMTS $S$,
$T$ for which $S\wedge T$ is defined, we have $S\wedge T\le_m S$ and
$S\wedge T\le_m T$, and if $S$ or $T$ is deterministic and $U$ is a SMTS
for which $U\le_m S$ and $U\le_m T$, then $S\wedge T$ is defined and
$U\le_m S\wedge T$.  We again note that this property implies
uniqueness, up to $\equiv_m$, of conjunction: if conjunction exists, it
must be given as above.

To generalize this to a quantitative greatest lower bound property, we
shall have reason to consider two different properties of the label
operator $\owedge$.  The first is analogous to the one for structural
composition above: we say that $\owedge$ is \emph{bounded} by a function
$C: \LL\times \LL\to \LL$ if $C$ is monotone in both coordinates, has
$C( \bot_\LL, \bot_\LL)= \bot_\LL$, $C( \alpha, \top_\LL)= C( \top_\LL,
\alpha)= \top_\LL$ for all $\alpha\in \LL$, and if it holds for all $k,
\ell, m\in \Spec$ for which $d_T( m, k)\ne \infty$ and $d_T( m,
\ell)\ne \infty$ that $k\owedge \ell$ is defined and
\begin{equation*}
  F( m, k\owedge \ell, C( \alpha, \alpha'))\sqsubseteq_\LL C( F( m, k,
  \alpha), F( m, \ell, \alpha'))
\end{equation*}
for all $\alpha, \alpha'\in \LL$.  For such a bounded $\owedge$ it can
be shown~\cite{conf/csr/BauerFLT12} that if $S$, $T$, $U$ are SMTS of
which $S$ or $T$ is deterministic, and if $h_m( U, S)\ne \top_\LL$ and
$h_m( U, T)\ne \top_\LL$, then $S\wedge T$ is defined and $h_m( U,
S\wedge T)\sqsubseteq_\LL C( h_m( U, S), h_m( U, T))$.

For the second, \emph{relaxed} boundedness property of $\owedge$, we
have to first introduce a notion of \emph{quantitative widening}.  For
$\alpha\in \LL$ and SMTS $S$, $T$, we say that $T$ is an
\emph{$\alpha$-widening} of $S$ if there is a relation $R\subseteq
S\times T$ for which $( s_0, t_0)\in R$ and such that for all $( s,
t)\in R$, $s\mayto[ k]_S s'$ if and only if $t\mayto[ \ell]_T t'$, and
$s\mustto[ k]_S s'$ if and only if $t\mustto[ \ell]_T t'$, for
$k\sqsubseteq_\Spec \ell$, $d( \ell, k)\sqsubseteq_\LL \alpha$, and $(
s', t')\in R$.  Thus up to unweighted two-sided refinement, $T$ is the
same as $S$, but transition labels in $T$ can be $\alpha$ ``wider'' than
in $S$.  (Hence also $S\le_m T$, but nothing general can be said about
quantitative refinement from $T$ to $S$,
\cf~\cite{conf/csr/BauerFLT12}.)

We say that the operator $\owedge$ is \emph{relaxed bounded} by a
function family $C=\{ C_{ \beta, \gamma}: \LL\times \LL\to \LL\mid
\beta, \gamma\in \LL\}$ if all $C_{ \beta, \gamma}$ are monotone in both
coordinates, have $C_{ \beta, \gamma}( \bot_\LL, \bot_\LL)= \bot_\LL$,
$C_{ \beta, \gamma}( \alpha, \top_\LL)= C_{ \beta, \gamma}( \top_\LL,
\alpha)= \top_\LL$ for all $\alpha\in \LL$, and if it holds for all $k,
\ell\in \Spec$ for which there is $m\in \Spec$ with $h_T( m, k)\ne
\top_\LL$ and $h_T( m, \ell)\ne \top_\LL$ that there exist $k', \ell'\in
\Spec$ with $k\sqsubseteq_\Spec k'$, $\ell\sqsubseteq_\Spec \ell'$,
$h_T( k', k)= \beta\ne \top_\LL$, and $h_T( \ell', \ell)= \gamma\ne
\top_\LL$, such that $k'\owedge \ell'$ is defined, and then for all
$m\in \Spec$ with $h_T( m, k)\ne \top_\LL$ and $d_T( m, \ell)\ne
\top_\LL$,
\begin{equation*}
  F( m, k'\owedge \ell', C_{ \beta, \gamma}( \alpha,
  \alpha'))\sqsubseteq_\LL C_{ \beta, \gamma}( F( m, k, \alpha), F( m,
  \ell, \alpha'))
\end{equation*}
for all $\alpha, \alpha'\in \LL$.  The following property can then be
shown~\cite[Thm.~5]{conf/csr/BauerFLT12}: Let $S$, $T$ be SMTS with $S$ or $T$
deterministic.  If there is an SMTS $U$ for which $h_m( U,
S)\ne\top_\LL$ and $h_m( U, T)\ne \top_\LL$, then there exist $\beta$-
and $\gamma$-widenings $S'$ of $S$ and $T'$ of $T$ for which $S'\wedge
T'$ is defined, and such that $h_m( U, S'\wedge T')\sqsubseteq_\LL C_{
  \beta, \gamma}( h_m( U, S), h_m( U, T))$ for all SMTS $U$ for which
$h_m( U, S)\ne \top_\LL$ and $h_m( U, T)\ne \top_\LL$.

\section{Robust Semantics of Modal Event-Clock Specifications}
\label{se:mecs}

As an application of the framework laid out in this paper, we consider
the modal event-clock specifications (MECS)
of~\cite{DBLP:conf/icfem/BertrandLPR09} and give them a robust semantics
as SMTS.  We choose MECS instead of a more expressive real-time
formalism such as \eg~timed automata~\cite{DBLP:journals/tcs/AlurD94}
mainly for ease of exposition; it is certainly possible to extend the
work presented here also to these formalisms.

We assume a fixed finite alphabet $\Sigma$ and let $\delta\notin \Sigma$
denote a special symbol which signifies passage of time.  Let $\Phi(
\Sigma)$ denote the set of closed clock constraints over $\Sigma$, given by
\begin{equation*}
  \Phi( \Sigma) \ni \phi ::= a\le k\mid a\ge k\mid \phi_1\wedge \phi_2
  \qquad( a\in \Sigma, k\in \Nat, \phi_1, \phi_2\in \Phi( \Sigma))\,.
\end{equation*}
A (real) clock valuation is a mapping $u: \Sigma\to \Realnn$; we say
that $u\models \phi$, for $\phi\in \Phi( \Sigma)$, if $u( a)$
satisfies $\phi$ for all $a\in \Sigma$, and we let $\llbracket
\phi\rrbracket=\{ u: \Sigma\to \Realnn\mid u\models \phi\}$.  For
$d\in \Realnn$ and $b\in \Sigma$ we define the valuations $u+
d= \lambda a.( u( a)+ d)$ and $u[ b]= \lambda a.( \text{if }
a= b\text{ then } 0\text{ else } u( a))$.  Note that for brevity we use
lambda notation for anonymous functions here.

We denote by $\I=\{[ x, y]\mid x\in \Realnn, y\in \Realnn\cup\{
\infty\}, x\le y\}$ the set of closed extended non-negative real
intervals, and define addition of intervals by $[ l, r]+[ l', r']=[ l+
l', r+ r']$.  An \emph{interval clock valuation} is a mapping $v:
\Sigma\to \I$ associating with each symbol $a$ a non-negative interval
$v( a)=[ l_a, r_a]\in \I$ of possible clock values.  We say that
$v\models \phi$, for $\phi\in \Phi( \Sigma)$, if there exists $u:
\Sigma\to \Realnn$ for which $u( a)\in v( a)$ for all $a\in \Sigma$ and
$u\models \phi$.  For $d\in \I$ and $b\in \Sigma$ we define $v+ d=
\lambda a.( v( a)+[ d, d])$ and $u[ b]= \lambda a.( \text{if } a=
b\text{ then }[ 0, 0]\text{ else } u( a))$.

A \emph{modal event-clock specification}
(MECS)~\cite{DBLP:conf/icfem/BertrandLPR09} is a tuple $A=( Q, q_0,
\mmayto_A, \mmustto_A)$ consisting of a finite set $Q$ of locations,
with initial location $q_0\in Q$, and \may and \must edges $\mmayto_A,
\mmustto_A\subseteq Q\times \Sigma\times \Phi( \Sigma)\times Q$ which
satisfy that for all $( q, a, g, q')\in \mmustto_A$ there exists $( q,
a, g', q')\in \mmayto_A$ with $\llbracket g\rrbracket\subseteq
\llbracket g'\rrbracket$.  As before we write $q\mayto[ a, g]_A q'$
instead of $( q, a, g, q')\in \mmayto_A$, similarly for $\mmustto_A$.
Figure~\ref{fi:mecs} shows some examples of MECS.

To facilitate robust analysis of MECS, we give their semantics not as
usual timed transition systems~\cite{DBLP:journals/tcs/AlurD94} (or as
modal region automata as in~\cite{DBLP:conf/icfem/BertrandLPR09}), but
as \emph{interval timed modal transition systems} (ITMTS).  These are
SMTS over
\begin{equation*}
  \Spec=( \Sigma\times\{[ 0, 0]\})\cup(\{ \delta\}\times \I)\subseteq(
  \Sigma\cup\{ \delta\})\times \I,
\end{equation*}
with $( a,[ l, r])\sqsubseteq_\Spec( a',[ l', r'])$ if and only if $a=
a'$, $l\ge l'$, and $r\le r'$ (hence $[ l, r]\subseteq[ l', r']$), and
thus with $\Imp= \Sigma\times\{ 0\}\cup\{ \delta\}\times \Realnn$.
Hence an implementation is a usual timed transition system, with
discrete transitions $s\mustto[ a, 0] s'$ and delay transitions
$s\mustto[ \delta, d] s'$.

The \emph{semantics} of a MECS $A=( Q, q_0, \mmayto_A, \mmustto_A)$ is
the ITMTS $\llparenthesis A\rrparenthesis=( S, s_0, \mayto_S,
\mustto_S)$ given as follows:
\begin{align*}
  S &= \{( q, v)\mid q\in Q, v: \Sigma\to \I\} \qquad s_0=( q_0,
  \lambda x. 0) \\
  \mmustto_S &= \{( q, v)\mustto[ a, 0]_S( q', v')\mid q\mustto[ a,
  g]_A q', v\models g, v'= v[ a]\} %\\
  %&\;
 \cup\{( q, v)\mustto[ \delta,{[ l, r]}]_S( q, v')\mid v'= v+[ l,
  r]\} \\
  \mmayto_S &= \{( q, v)\mayto[ a, 0]_S( q', v')\mid q\mayto[ a, g]_A
  q', v\models g, v'= v[ a]\} %\\
  %&\;
 \cup\{( q, v)\mayto[ \delta,{[ l, r]}]_S( q, v')\mid v'= v+[ l,
  r]\}
\end{align*}

Note that the ``real'', precise semantics of $A$ as a timed transition
system~\cite{DBLP:journals/tcs/AlurD94} is an implementation of
$\llparenthesis A\rrparenthesis$, also any of the ``relaxed'' or
``robust'' semantics
of~\cite{DBLP:conf/concur/BouyerLMST11,DBLP:conf/hybrid/GuptaHJ97,DBLP:conf/time/SwaminathanF07,DBLP:conf/ifipTCS/SwaminathanFK08}
are implementations of $\llparenthesis A\rrparenthesis$; any robust
semantics ``lives'' in our framework.  As we are using closed clock
constraints for MECS, $\llparenthesis A\rrparenthesis$ as defined above
is compactly branching.

Refinement of MECS is defined semantically: $A\le_m B$ if
$\llparenthesis A\rrparenthesis\le_m \llparenthesis B\rrparenthesis$.
Note that the refinement of~\cite{DBLP:conf/icfem/BertrandLPR09} is
different (indeed it is not quantitative in our sense).  By definition
of modal refinement, a specification $S\le_m \llparenthesis
A\rrparenthesis$ is a \emph{more precise}, or less relaxed,
specification of the semantics of $A$: any delay intervals on
transitions $s\mayto[ \delta,{[ l, r]}]_S s'$ are contained in intervals
$t\mayto[ \delta,{[ l', r']}]_{ \llparenthesis A\rrparenthesis} t'$ (and
similarly for \must transitions).

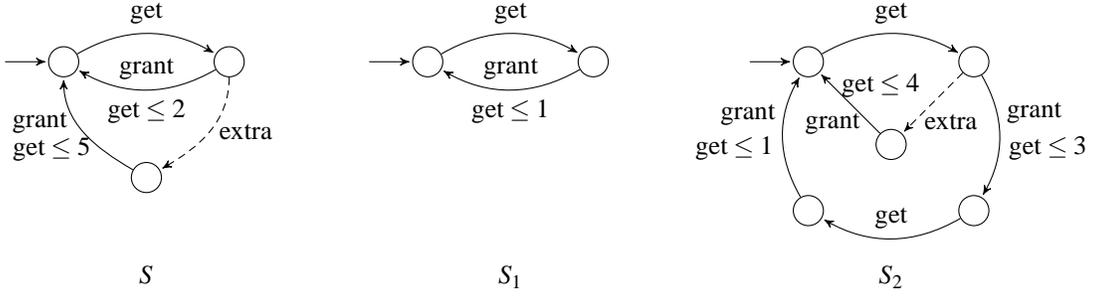
\begin{figure}[tp]
  \centering
  \begin{tikzpicture}[->,>=stealth',shorten >=1pt,auto,node
    distance=2.0cm,initial text=,scale=1.1]
    \tikzstyle{every node}=[font=\small] 
    \tikzstyle{every state}=[fill=white,shape=circle,inner
    sep=.5mm,minimum size=4mm]
    \begin{scope}
      \node[state,initial] (1) at (0,0) {};
      \node[state] (2) at (2,0) {};
      \node[state] (3) at (1,-1.4) {};
      \node at (1,-2.6) {$S_{ \vphantom{1}}$};
      \path (1) edge [out=30,in=150] node [above] {get} (2);
      \path (2) edge [out=210,in=-30] node [above] {grant} node [below]
      {$\text{get}\le 2$} (1);
      \path (2) edge [densely dashed,out=270,in=30] node [right] {extra}
      (3);
      \path (3) edge [out=150,in=270] node [left,pos=.6] {grant}
      node[left,pos=.3] {$\text{get}\le 5$} (1);
    \end{scope}
    \begin{scope}[xshift=4.4cm]
      \node[state,initial] (1) at (0,0) {};
      \node[state] (2) at (2,0) {};
      \node at (1,-2.6) {$S_1$};
      \path (1) edge [out=30,in=150] node [above] {get} (2);
      \path (2) edge [out=210,in=-30] node [above] {grant} node [below]
      {$\text{get}\le 1$} (1);
    \end{scope}
    \begin{scope}[xshift=9cm]
      \node[state,initial] (1) at (0,0) {};
      \node[state] (2) at (2,0) {};
      \node[state] (3) at (1,-1) {};
      \node[state] (1') at (2,-1.8) {};
      \node[state] (2') at (0,-1.8) {};
      \node at (1,-2.6) {$S_2$};
      \path (1) edge [out=30,in=150] node [above] {get} (2);
      \path (2) edge [densely dashed] node [right,pos=.8] {extra} (3);
      \path (3) edge node [left,pos=.17] {grant} node [right,pos=.8]
      {$\text{get}\le 4$} (1);
      \path (2) edge [out=-60,in=60] node [right,pos=.3] {grant}
      node [right,pos=.6] {$\text{get}\le 3$} (1');
      \path (1') edge [out=210,in=330] node [above] {get} (2');
      \path (2') edge [out=120,in=240] node[left,pos=.68] {grant} node
      [left,pos=.4] {$\text{get}\le 1$} (1);
    \end{scope}
  \end{tikzpicture}
  \caption{%
    \label{fi:mecs}
    An MECS model $S$ of a resource specification,
    \cf~\cite{DBLP:conf/icfem/BertrandLPR09}, and two refinement
    candidates $S_1$, $S_2$.  As customary, we omit \may-transitions
    which have an underlying \must-transition with the same label.  Note
    that $S_1\le_m S$ and $S_2\not\le_m S$, but $d_m( S_2, S)= 1$.}
\end{figure}

We are interested in \emph{timing differences} of (refinements of) MECS,
\ie~in expressing how much two ITMTS can differ in the timings of their
behaviors.  Given two finite traces $\sigma=( a_0, x_0),\dots,( a_n,
x_n)$ and $\sigma'=( a_0, x_0'),\dots,( a_n, x_n')$ (note that the
discrete labels in $\Sigma\cup\{ \delta\}$ are the same), their timing
difference is $|( x_0+ x_1+\dots+ x_n)-( x_0'+ x_1'+\dots+ x_n')|$, and
what interests us is the \emph{maximal} timing difference at any point
of the runs.  Hence we want the distance between $\sigma$ and $\sigma'$
to be $\max_{ m= 0,\dots, n}| \sum_{ i= 0}^m x_i- \sum_{ i= 0}^m x_i'|$,
and with the $\max_{ m= 0,\dots, n}$ replaced by $\sup_{ m\in \Nat}$ for
infinite traces.  This is precisely the \emph{maximum-lead distance}
of~\cite{DBLP:conf/formats/HenzingerMP05,journals/jlap/ThraneFL10}, and
we show below how it fits in the framework of this paper.

Note that the accumulating distance
of~\cite{DBLP:conf/mfcs/BauerFJLLT11} measures something entirely
different: for the finite traces above, it is $| x_0- x_0'|+ \discount|
x_1- x_1'|+\dots+ \discount^n| x_n- x_n'|$, hence measuring the sum of
the differences in the individual timings of transitions rather than the
overall timing difference.  Thus the work laid out
in~\cite{DBLP:conf/mfcs/BauerFJLLT11} is not applicable to our setting,
showing the strength of the more general approach
of~\cite{conf/csr/BauerFLT12}.

Let $\LL=( \Realnn\cup\{ \infty\})^\Real$, the set of mappings from
\emph{leads} to distances, define $F: \Imp\times \Imp\times \LL\to \LL$ by
\begin{equation*}
  F(( a, t),( a', t'), \alpha)=
  \begin{cases}[c]
    \top_\LL &\text{if } a\ne a'\,, \\
    \lambda d. \max(| d+ t- t'|, \alpha( d+ t- t'))
    &\text{if } a= a'\,
  \end{cases}
\end{equation*}
and extend $F$ to specifications by $F( k, \ell, \alpha)= \sup_{ m\in
  \llbracket k\rrbracket} \inf_{ n\in \llbracket \ell\rrbracket} F( m,
n, \alpha)$.  Define $g: \LL\to \Realnn\cup\{ \infty\}$ by $g( \alpha)=
\alpha( 0)$; the maximum-lead distance assuming the lead is zero.  Using
our characterization of $h_T$ from~\eqref{eq:trdist}, it can then be
shown that $d_T= g\circ h_T: \Specseq\times \Specseq\to \Realnn\cup\{
\infty\}$ is precisely the maximum-lead distance,
\cf~\cite{DBLP:conf/formats/HenzingerMP05,DBLP:journals/corr/abs-1107-1205}.  We
also instantiate our definitions of modal and thorough refinement
distance for ITMTS; for MECS $A$, $B$ we let $d_m( A, B)=
d_m(\llparenthesis A\rrparenthesis, \llparenthesis B\rrparenthesis)$,
$d_t( A, B)= d_t(\llparenthesis A\rrparenthesis, \llparenthesis
B\rrparenthesis)$.

\emph{Determinism} for ITMTS is the same as
in~\cite{DBLP:conf/mfcs/BauerFJLLT11}: if $k_1, k_2\in \Spec$, with
$k_1=( a_1,[ l_1, r_1])$, $k_2=( a_2,[ l_2, r_2])$, then there is $k\in
\Spec$ with $h_T( k, k_1)\ne \top_\LL$ and $h_T( k, k_2)\ne \top_\LL$ if and
only if $a_1= a_2$.  Hence an ITMTS $S$ is deterministic if and only if
it holds for all $s\in S$ that $s\mayto[{( a,[ l_1, r_1])}]_S s_1$ and
$s\mayto[{( a,[ l_2, r_2])}]_S s_2$ imply $[ l_1, r_1]=[ l_2, r_2]$ and
$s_1= s_2$.  For an MECS $A$, $\llparenthesis A\rrparenthesis$ is hence
deterministic if and only if for all locations $q$, $q\mayto[ a, g_1]
q_1$ and $q\mayto[ a, g_2] q_2$ imply that $\llbracket g_1\rrbracket=
\llbracket g_2\rrbracket$ and $q_1= q_2$.  This is a stronger notion of
determinism than in~\cite{DBLP:conf/icfem/BertrandLPR09}; we will call
it \emph{strong determinism} for differentiation.

For \emph{structural composition} of ITMTS we use CSP-style
synchronization on discrete labels and intersection of intervals.  Note
that this is different from~\cite{DBLP:conf/mfcs/BauerFJLLT11} which
instead uses addition of intervals.  Given $( a,[ l, r]),( a',[ l',
r'])\in \Spec$ we hence define
\begin{equation*}
  ( a,[ l, r])\obar( a',[ l', r'])=
  \begin{cases}
    ( a,[ \max( l, l'), \min( r, r')])  %\hspace*{-5em} & \\
    &\text{if } a= a' \text{ and }
    \max( l, l')\le \min(
    r, r')\,, \\
    \text{undefined} &\text{otherwise}\,.
  \end{cases}
\end{equation*}

It can be shown that $\obar$ is bounded by $P( \alpha, \alpha')= \max(
\alpha, \alpha')$.  Also, the notion of structural composition of ITMTS
we obtain is consistent with the one of synchronized product
of~\cite{DBLP:conf/icfem/BertrandLPR09} (denoted $\otimes$ in that
paper).  Figure~\ref{fi:struct} depicts some examples of structural
compositions.

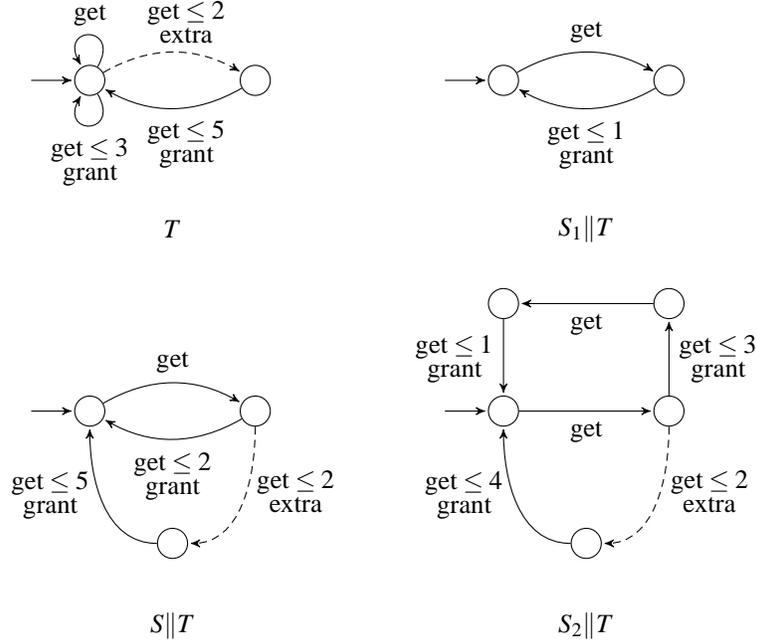
\begin{figure}[tp]
  \centering
  \begin{tikzpicture}[->,>=stealth',shorten >=1pt,auto,node
    distance=2.0cm,initial text=,scale=1.1]
    \tikzstyle{every node}=[font=\small] 
    \tikzstyle{every state}=[fill=white,shape=circle,inner
    sep=.5mm,minimum size=4mm]
    \begin{scope}
      \node[state,initial] (1) at (0,0) {};
      \node[state] (2) at (2,0) {};
      \node at (1,-1.8) {$T$};
      \path (1) edge [out=60,in=120,loop] node [above] {get} (1);
      \path (1) edge [out=-60,in=-120,loop] node [below]
      {$\stackrel{ \textstyle \text{get}\le 3}{ \text{grant}}$} (1);
      \path (1) edge [densely dashed,out=30,in=150] node [above,pos=.6]
      {$\stackrel{ \textstyle \text{get}\le 2}{ \text{extra}}$} (2);
      \path (2) edge [out=210,in=-30] node [below,pos=.4] {$\stackrel{
          \textstyle \text{get}\le 5}{ \text{grant}}$} (1);
    \end{scope}
    \begin{scope}[xshift=5cm]
      \node [state,initial] (1) at (0,0) {};
      \node [state] (2) at (2,0) {};
      \node at (1,-1.8) {$S_1\| T$};
      \path (1) edge [out=30,in=150] node [above] {get} (2);
      \path (2) edge [out=210,in=-30] node [below] {$\stackrel{
          \textstyle \text{get}\le 1}{ \text{grant}}$} (1);
    \end{scope}
    \begin{scope}[yshift=-4cm]
      \node [state,initial] (1) at (0,0) {};
      \node [state] (2) at (2,0) {};
      \node [state] (3) at (1,-1.6) {};
      \node at (1,-2.6) {$S\| T$};
      \path (1) edge [out=30,in=150] node [above] {get} (2);
      \path (2) edge [out=210,in=-30] node [below] {$\stackrel{
          \textstyle \text{get}\le 2}{ \text{grant}}$} (1);
      \path (2) edge [densely dashed,out=270,in=0] node [right,pos=.4]
      {$\stackrel{ \textstyle \text{get}\le 2}{ \text{extra}}$} (3);
      \path (3) edge [out=180,in=270] node [left,pos=.58] {$\stackrel{
          \textstyle \text{get}\le 5}{ \text{grant}}$} (1);
    \end{scope}
    \begin{scope}[yshift=-4cm,xshift=5cm]
      \node [state,initial] (1) at (0,0) {};
      \node [state] (2) at (2,0) {};
      \node [state] (3) at (2,1.3) {};
      \node [state] (4) at (0,1.3) {};
      \node [state] (5) at (1,-1.6) {};
      \node at (1,-2.6) {$S_2\| T$};
      \path (1) edge node [below] {get} (2);
      \path (2) edge node [right] {$\stackrel{ \textstyle \text{get}\le
          3}{ \text{grant}}$} (3);
      \path (3) edge node [below] {get} (4);
      \path (4) edge node [left] {$\stackrel{ \textstyle \text{get}\le
          1}{ \text{grant}}$} (1);
      \path (2) edge [densely dashed,out=270,in=0] node [right,pos=.4]
      {$\stackrel{ \textstyle \text{get}\le 2}{ \text{extra}}$} (5);
      \path (5) edge [out=180,in=270] node [left,pos=.58] {$\stackrel{
          \textstyle \text{get}\le 4}{ \text{grant}}$} (1);
    \end{scope}
  \end{tikzpicture}
  \caption{%
    \label{fi:struct}
    A MECS model $T$ of a process accessing the resource $S$ from
    Fig.~\ref{fi:mecs}, together with the structural compositions $S\|
    T$, $S_1\| T$, and $S_2\| T$.  Note that $d_m( S_2\| T, S\| T)= 1$.}
\end{figure}

\begin{theorem}
  Let $A$, $B$, $A'$, $B'$ be MECS.  With $\|$ the notion of
  synchronized product of MECS
  from~\cite{DBLP:conf/icfem/BertrandLPR09}, $\llparenthesis A\|
  B\rrparenthesis\equiv_m \llparenthesis A\rrparenthesis\|\llparenthesis
  B\rrparenthesis$.  Additionally, $d_m( A\| A', B\| B')\le \max( d_m(A,
  B), d_m( A', B'))$.
\end{theorem}

\begin{proof}
  $\llparenthesis A\| B\rrparenthesis\equiv_m \llparenthesis
  A\rrparenthesis\|\llparenthesis B\rrparenthesis$ is clear from the
  definitions.  For the second part, we have $h_m( A\| A', B\|
  B')\sqsubseteq_\LL P( h_m( A, B), h_m( A', B'))= \max( h_m( A, B),
  h_m( A', B'))$ by~\cite[Thm.~2]{conf/csr/BauerFLT12}, and as $g:
  \LL\to \Realnn\cup\{ \infty\}$ is a homomorphism, the claim follows.
\end{proof}

For \emph{quotient} of ITMTS we define, for labels $( a,[ l, r]),( a',[
l', r'])\in \Spec$,
\begin{equation*}
  ( a',[ l', r'])\obslash( a,[ l, r])=
  \begin{cases}
    \text{undefined} &\text{if } a\ne a'\,, \\
    ( a,[ l', \infty]) &\text{if } a= a'\text{ and } l< l'\le r\le r'\,,
    \\
    ( a,[ l', r']) &\text{if } a= a'\text{ and } l< l'\le r'< r\,, \\
    \text{undefined} &\text{if } a= a'\text{ and } l\le r< l'\le r'\,, \\
    ( a,[ 0, \infty]) &\text{if } a= a'\text{ and } l'\le l\le
    r\le r'\,, \\
    ( a,[ 0, r']) &\text{if } a= a'\text{ and } l'\le l\le r< r'\,,
    \\
    \text{undefined} &\text{if } a= a'\text{ and } l'\le r'< l\le r\,.
  \end{cases}
\end{equation*}
The intuition is that to obtain the maximal solution $[ p, q]$ to an
equation $[ l, r]\obar[ p, q]\sqsubseteq_\Spec[ l', r']$, whether $p$
and $q$ must restrain the interval in the intersection, or can be $0$
and $\infty$, respectively, depends on the position of $[ l, r]$
relative to $[ l', r']$, \cf~Figure~\ref{fi:quotient-mecs}.
% if the intervals are disjoint, the intersection $[ l, r]\obar[ p, q]$
% must be empty and hence quotient can be undefined.
It can be shown that the operator $\obslash$ is quantitatively
well-behaved.

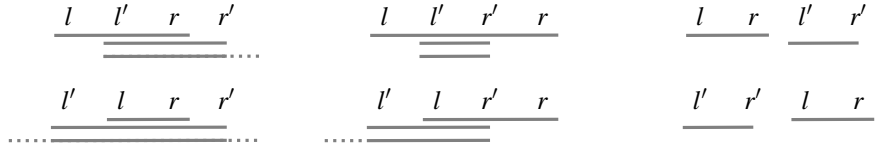
\begin{figure}[tp]
  \centering
  \begin{tikzpicture}[scale=.7]
    \tikzstyle{every node}=[font=\small,anchor=base]
    \tikzset{rangebar/.style={color=gray,very thick}}

    \begin{scope}
      \node (l) at (0,0) {$l$};
    \node (l') at (1,0) {$l'$};
    \node (r) at (2,0) {$r$};
    \node (r') at (3,0) {$r'$};
    \path[rangebar] (l.south west) edge (r.south east);
    \path[rangebar] ($(l'.south west)+(0,-1.5mm)$) edge ($(r'.south)+(0,-1.5mm)$);
    \path[rangebar] ($(l'.south west)+(0,-4mm)$) edge ($(r'.south)+(0mm,-4mm)$)
    edge [dotted] ($(r'.south west)+(10mm,-4mm)$);
    \end{scope}
    
    \begin{scope}[xshift=6cm]
      \node (l) at (0,0) {$l$};
    \node (l') at (1,0) {$l'$};
    \node (r') at (2,0) {$r'$};
    \node (r) at (3,0) {$r$};
    \path[rangebar] (l.south west) edge (r.south east);
    \path[rangebar] ($(l'.south west)+(0,-1.5mm)$) edge ($(r'.south)+(0,-1.5mm)$);
    \path[rangebar] ($(l'.south west)+(0,-4mm)$) edge ($(r'.south)+(0mm,-4mm)$);
    \end{scope}

    \begin{scope}[xshift=12cm]
      \node (l) at (0,0) {$l$};
    \node (r) at (1,0) {$r$};
    \node (l') at (2,0) {$l'$};
    \node (r') at (3,0) {$r'$};
    \path[rangebar] (l.south west) edge (r.south east);
    \path[rangebar] ($(l'.south west)+(0,-1.5mm)$) edge ($(r'.south)+(0,-1.5mm)$);
    \end{scope}

    \begin{scope}[yshift=-1.6cm]
      \begin{scope}
        \node (l') at (0,0) {$l'$};
        \node (l) at (1,0) {$l$};
        \node (r) at (2,0) {$r$};
        \node (r') at (3,0) {$r'$};
        \path[rangebar] (l.south west) edge (r.south east);
        \path[rangebar] ($(l'.south west)+(0,-1.5mm)$) edge
        ($(r'.south)+(0,-1.5mm)$);
        \path[rangebar] ($(l'.south
        west)+(-8mm,-4mm)$)  edge[dotted]  ($(l'.south
        west)+(0mm,-4mm)$);
        \path[rangebar]   ($(l'.south
        west)+(0mm,-4mm)$) edge ($(r'.south)+(0mm,-4mm)$) edge [dotted]
        ($(r'.south west)+(10mm,-4mm)$);
      \end{scope}
    
    \begin{scope}[xshift=6cm]
      \node (l') at (0,0) {$l'$}; \node (l) at (1,0) {$l$}; \node (r')
      at (2,0) {$r'$}; \node (r) at (3,0) {$r$}; \path[rangebar]
      (l.south west) edge (r.south east); \path[rangebar] ($(l'.south
      west)+(0,-1.5mm)$) edge ($(r'.south)+(0,-1.5mm)$);
      \path[rangebar] ($(l'.south west)+(-8mm,-4mm)$) edge[dotted] ($(l'.south west)+(0mm,-4mm)$); 
      \path[rangebar] ($(l'.south west)+(0mm,-4mm)$) edge
      ($(r'.south)+(0mm,-4mm)$);
    \end{scope}

    \begin{scope}[xshift=12cm]
      \node (l') at (0,0) {$l'$}; \node (r') at (1,0) {$r'$}; \node (l)
      at (2,0) {$l$}; \node (r) at (3,0) {$r$}; \path[rangebar]
      (l.south west) edge (r.south east); \path[rangebar] ($(l'.south
      west)+(0,-1.5mm)$) edge ($(r'.south)+(0,-1.5mm)$);
    \end{scope}
  \end{scope}

  \end{tikzpicture}
  \caption{%
    \label{fi:quotient-mecs}
    Quotient $[ l', r']\obslash[ l, r]$ of intervals, six cases.  Top
    bar: $[ l, r]$; middle bar: $[ l', r']$; bottom bar: quotient.  Note
    that for the two cases on the right, quotient is undefined.
  }
\end{figure}

We can lift our quotient from the semantic ITMTS level to MECS as
follows: A clock constraint in $\Phi( \Sigma)$ is equivalent to a
mapping $\Sigma\to \mathbbm J$, where $\mathbbm J=\{[ x, y]\mid x\in
\Nat, y\in \Nat\cup\{ \infty\}, x\le y\}\subseteq \I$ denotes the set of
closed extended non-negative integer intervals, and then we can define
$\phi'\obslash \phi= \lambda a.( \phi'( a)\obslash \phi( a))$ with
$\obslash$ defined on intervals as above.  Our quotient of MECS is then
defined as in~\cite{DBLP:conf/icfem/BertrandLPR09}, but with their guard
operation replaced by our $\obslash$ (hence our quotient is different
from theirs, which is to be expected as the notions of refinement are
different).

\begin{theorem}
  Let $A$, $B$, $X$ be MECS for which $B\bbslash A$ exists, then
  $\llparenthesis B\bbslash A\rrparenthesis\equiv \llparenthesis
  B\rrparenthesis\bbslash \llparenthesis A\rrparenthesis$.  If $A$ is
  strongly deterministic, then
  $d_m( X, B\bbslash A)\le d_m( A\| X, B)$, and $X\le_m B\bbslash A$ if
  and only if $A\| X\le_m B$.
\end{theorem}

\begin{proof}
  $\llparenthesis B\bbslash A\rrparenthesis\equiv \llparenthesis
  B\rrparenthesis\bbslash \llparenthesis A\rrparenthesis$ is clear from
  the definitions.  For the second part, $X\le_m B\bbslash A$ if
  and only if $A\| X\le_m B$ by~\cite[Thm.~3]{conf/csr/BauerFLT12}, and
  by the same theorem, $h_m( X, B\bbslash A)\sqsubseteq h_m( A\| X, B)$,
  so as $g: \LL\to \Realnn\cup\{ \infty\}$, the claim follows.
\end{proof}

%\paragraph{Conjunction.}

The \emph{conjunction} operator on labels of ITMTS is defined using
intersection of intervals like for structural composition, hence we let
$k\owedge \ell= k\obar \ell$ for $k, \ell\in \Spec$.  The intuition is
that transition intervals give constraints on timings; hence a
synchronized transition has to satisfy both interval constraints.  It
can be shown that $\owedge$ is \emph{not} bounded, but relaxed bounded
by $C_{ \beta, \gamma}( \alpha, \alpha')= \max( \alpha,
\alpha')\oplus_\LL \max( \beta, \gamma)$.

Our notion of conjunction is consistent with the one for MECS
in~\cite{DBLP:conf/icfem/BertrandLPR09}, and to make use of relaxed
boundedness, we need to lift the notion of quantitative widening from
the semantic ITMTS level to MECS.  This is done by defining, for a clock
constraint $\phi: \Sigma\to \mathbbm J$ and $n\in \Nat$, the
$n$-extended constraint $\phi_{ +n}= \lambda a. \phi( a)+[ -n, n]$ (this
is similar to a construction in~\cite{DBLP:conf/concur/BouyerLMST11}),
and then saying that a MECS $B$ is an $n$-widening of an MECS $A$ if
there is a relation $R\subseteq Q_A\times Q_B$ for which $( q_0^A,
q_0^B)\in R$, and for all $( q_A, q_B)\in R$, $q_A\mayto[ a, g]_A q_A'$
if and only if $q_B\mayto[ a, g_{ +n}] q_B'$ with $( q_B, q_B')\in R$
and similarly for \must~transitions.

\begin{theorem}
  Let $A$, $B$ be MECS.  With $\wedge$ the notion of greatest lower
  bound from~\cite{DBLP:conf/icfem/BertrandLPR09}, $\llparenthesis
  A\wedge B\rrparenthesis\equiv \llparenthesis A\rrparenthesis\wedge
  \llparenthesis B\rrparenthesis$.  If $A$ or $B$ is strongly
  deterministic and there is a MECS $C$ for which $d_m( C, A)\ne \infty$
  and $d_m( C, B)\ne \infty$, then there are an $n$-widening $A'$ of $A$
  and an $m$-widening $B'$ of $B$ for which $A'\wedge B'$ is defined,
  and such that $d_m( C, A'\wedge B')\le \max( d_m( C, A), d_m( C, B))+
  \max( n, m)$ for all MECS $C$ for which $d_m( C, A)\ne \infty$ and
  $d_m( C, B)\ne \infty$.
\end{theorem}

\begin{proof}
  $\llparenthesis A\wedge B\rrparenthesis\equiv \llparenthesis
  A\rrparenthesis\wedge \llparenthesis B\rrparenthesis$ by definition,
  and the second claim follows from~\cite[Thm.~5]{conf/csr/BauerFLT12}
  and the homomorphism property of $g: \LL\to \Realnn\cup\{ \infty\}$.
\end{proof}

%\section{Conclusion}

\bibliographystyle{eptcs}
\bibliography{wmbib}

\end{document}